%% file: main.tex
\documentclass[11pt]{article}
\usepackage{mathpazo}
\usepackage{fullpage}
\usepackage{algorithm}
\usepackage{algcompatible}
\usepackage{nicefrac}
\usepackage{amsmath}
\usepackage{amssymb}
\usepackage{amsthm}
\usepackage[nottoc,numbib]{tocbibind}
\usepackage[linktocpage=true,pagebackref=false]{hyperref}
\hypersetup{colorlinks=true,linkcolor=black,citecolor=blue}
\usepackage{tikz}
\tikzstyle{vertex}=[circle, draw, inner sep=0pt, minimum size=6pt]

\usepackage{caption}
\usetikzlibrary{patterns,positioning,decorations.pathmorphing,decorations.pathreplacing}
\usetikzlibrary{arrows,automata,snakes}
\usepackage{color}
\usepackage{xspace}

\newtheorem{theorem}{Theorem}
\newtheorem{definition}[theorem]{Definition}

\newtheorem{conjecture}[theorem]{Conjecture}

\newtheorem{corollary}[theorem]{Corollary}
\newtheorem{lemma}[theorem]{Lemma}

\makeatletter
\newtheorem*{rep@theorem}{\rep@title}
\newcommand{\newreptheorem}[2]{%
\newenvironment{rep#1}[1]{%
 \def\rep@title{#2 \ref{##1}}%
 \begin{rep@theorem}}%
 {\end{rep@theorem}}}
\makeatother

\newcommand{\closest}{\mbox{{\sc Closest Pair}}\xspace}

\newcommand{\bcp}{\mbox{{\sc BCP}}\xspace}
\newcommand{\ovec}{\mbox{{\sc OV}}\xspace}

\newcommand{\reals}{\mathbb{R}}
\newcommand{\integers}{\mathbb{N}}

\makeatletter
\newcommand\slarge{\@setfontsize\slarge{14}{14}}
\makeatother

\newcommand{\sph}{\mathsf{sph}}
\newcommand{\bisph}{\mathsf{bsph}}
\newcommand{\cd}{\mathsf{cd}}
\newcommand{\bicd}{\mathsf{bcd}}

\def\DEBUG{true}

\ifdefined\DEBUG

  \def\rem#1{{\marginpar{\raggedright\scriptsize #1}}}
   \newcommand{\karr}[1]{\rem{\textcolor{orange}{$\bullet$ #1}}}
   \newcommand{\bunr}[1]{\rem{\textcolor{blue}{$\bullet$ #1}}}
   \newcommand{\roer}[1]{\rem{\textcolor{green}{$\bullet$ #1}}}
   \newcommand{\mytodor}[1]{\rem{\textcolor{red}{$\bullet$ #1}}}

\else

  \newcommand{\mytodor}[1]}

  \newcommand{\karr}[1]{}
  \newcommand{\bunr}[1]{}
  \newcommand{\roer}[1]{}
\fi

\newcommand{\squishlist}{
 \begin{list}{$\bullet$}
  { \setlength{\itemsep}{0pt}
     \setlength{\parsep}{3pt}
     \setlength{\topsep}{3pt}
     \setlength{\partopsep}{0pt}
     \setlength{\leftmargin}{1.5em}
     \setlength{\labelwidth}{1em}
     \setlength{\labelsep}{0.5em} } }

\newcommand{\squishlisttwo}{
 \begin{list}{$\bullet$}
  { \setlength{\itemsep}{0pt}
     \setlength{\parsep}{0pt}
    \setlength{\topsep}{0pt}
    \setlength{\partopsep}{0pt}
    \setlength{\leftmargin}{2em}
    \setlength{\labelwidth}{1.5em}
    \setlength{\labelsep}{0.5em} } }

\newcommand{\squishend}{
  \end{list}  }

\begin{document}


\title{\textbf{On the Complexity of Closest Pair via \\
Polar-Pair of Point-Sets}}


\author{{Roee David}\\
Datorama, Israel\\
e-mail:\texttt{daroeevid@gmail.com}
\and
{Karthik C.\ S.}\\
Weizmann Institute of Science, Israel\\
e-mail:\texttt{karthik.srikanta@weizmann.ac.il}
\and 
{Bundit Laekhanukit}\\
Shanghai University of Finance and Economics, China\\ \& 
Max-Planck-Institut f\"{u}r Informatik, Saarbr\"{u}cken, Germany\\
e-mail:\texttt{blaekhan@mpi-inf.mpg.de}}


\date{}
\maketitle
\renewcommand{\thefootnote}{\fnsymbol{footnote}}
\input{Abstract}
\clearpage
\tableofcontents
\clearpage

\input{Introduction}
\input{Preliminaries}
\input{ell1}
\input{ell0}
\input{ellp_1to2}
\input{ellp}

\input{Connection}

\input{Conclusion}
\input{Acknowledgement}

\bibliographystyle{alpha}
\bibliography{geom}
\appendix
\clearpage
\input{ell2}

\end{document}

%% file: Abstract.tex

\begin{abstract}
Every graph $G$ can be represented by a collection of equi-radii spheres 
in a $d$-dimensional metric $\Delta$ such that 
there is an edge $uv$ in $G$ if and only if 
the spheres corresponding to $u$ and $v$ intersect.
The smallest integer $d$ such that $G$ can be represented by a 
collection of spheres (all of the same radius) in $\Delta$ is called the {\em sphericity} of $G$,
and if the collection of spheres are non-overlapping, then
the value $d$ is called the {\em contact-dimension} of $G$.
In this paper, we study the sphericity and contact dimension of the complete
bipartite graph $K_{n,n}$ in various $L^p$-metrics 
and consequently
connect the complexity of the monochromatic closest pair and bichromatic closest pair problems.
\end{abstract}

{\color{white}{\em Polar Bear}}

%% file: Introduction.tex
\section{Introduction}
\label{sec:intro}

%
%
This paper studies the geometric representation of a complete bipartite
 graph in $L^p$-metrics and consequently
connects the complexity of the closest pair and bichromatic closest pair problems beyond certain dimensions.
Given a point-set $P$ in a $d$-dimensional $L^p$-metric,
an {\em $\alpha$-distance graph} is a graph
$G=(V,E)$ with a vertex set $V=P$ and an edge set 
$$E=\{uv:\|u-v\|_p \leq \alpha; u,v\in P; u\neq v\}.$$
%
%
In other words, points in $P$ are centers of spheres of radius
$\alpha/2$, and $G$ has an edge $uv$ if and only if 
the spheres centered at $u$ and $v$ intersect.
The {\em sphericity} of a graph $G$ in an $L^p$-metric,
denoted by $\sph_p(G)$, is the smallest dimension $d$ such that $G$
is isomorphic to some $\alpha$-distance graph in a $d$-dimensional
$L^p$-metric, for some constant $\alpha>0$.
The sphericity of a graph in the $L^{\infty}$-metric is known as {\em cubicity}.
A notion closely related to sphericity is {\em contact-dimension},
which is defined in the same manner except that 
the spheres representing $G$ must be non-overlapping.
To be precise, an {\em $\alpha$-contact graph} $G=(V,E)$ of a point-set $P$ is 
an $\alpha$-distance graph of $P$ such that every edge $uv$ of $G$ 
has the same distance (i.e., $\|u-v\|_p=\alpha$).
Thus, $G$ has the vertex set $V=P$ and has an edge set $E$ such that 
$$
\forall uv\in E, \quad\|u-v\|_p=\alpha \quad\mbox{and}\quad\forall uv\not\in E,\quad
\|u-v\|_p>\alpha.
$$
The contact-dimension of a graph $G$ in the $L^p$-metric, 
denoted by $\cd_p(G)$, is the smallest integer $d \geq 1$ such that 
$G$ is isomorphic to a contact-graph in the $d$-dimensional $L^p$-metric.
We will use distance and contact graphs to mean 
$1$-distance and $1$-contact graphs.

We are interested in determining the sphericity and the contact-dimension
of the biclique $K_{n,n}$ in various $L^p$-metrics.
For notational convenience, we denote $\sph_p(K_{n,n})$ by
$\bisph(L^p)$,
the {\em biclique sphericity} of the $L^p$-metric,
and denote $\cd_p(K_{n,n})$  by
$\bicd(L^p)$,
the {\em biclique contact-dimension} of the $L^p$-metric.
We call a pair of point-sets $(A,B)$ \emph{polar} if it is 
the partition of the vertex set of a contact graph isomorphic to $K_{n,n}$. 
More precisely, a pair of point-sets $(A,B)$ is polar in an $L^p$-metric if
there exists a constant $\alpha>0$ such that every inner-pair $u,u'\in A$ (resp., $v,v'\in B$)
has $L^p$-distance greater than $\alpha$ while every crossing-pair $u\in A,v\in B$ has $L^p$-distance exactly $\alpha$.

The biclique sphericity and contact-dimension of the $L^2$ and $L^{\infty}$ metrics
are well-studied in literature (see \cite{R69cubicity,M84,M85-contact-pattern,FM88,M91,BL05}).
Maehara \cite{M91,M84} showed that $n < \bisph(L^2) \leq (1.5)n$,
and Maehara and Frankl \& Maehara \cite{M85-contact-pattern, FM88} showed that $(1.286)n-1 < \bicd(L^2) < (1.5)n$.
For cubicity, Roberts \cite{R69cubicity} showed that
$\bicd(L^{\infty})=\bisph(L^{\infty})=2\log_2{n}$. 
Nevertheless, for other $L^p$-metrics, contact dimension and sphericity are not well-studied.

\subsection{Our Results and Contributions}
\label{sec:results}

Our main conceptual contribution is connecting the complexity of the (monochromatic) closest pair problem (\closest) to that of the bichromatic closest pair problem (\bcp) through the contact dimension of the biclique. This is discussed in subsection~\ref{conceptual}. Our main technical contributions are bounds on the contact dimension and sphericity of the biclique for various $L^p$-metrics. This is discussed in subsection~\ref{technical}. Finally, as an application of the connection discussed in subsection~\ref{conceptual} and the bounds discussed in subsection~\ref{technical}, we show computational equivalence between monochromatic and bichromatic closest pair problems. 

\subsubsection{Connection between \closest and \bcp}\label{conceptual}\label{sec:connection}

In \closest, we are given a point-set of cardinality $m$ as input and our goal is to find a pair of distinct points in the set with minimum distance.
\bcp is a generalization of \closest, in which we are given as input a set of $m/2$ red points and a set of $m/2$ blue points,
and the goal is to find a pair of red-blue points (i.e., bichromatic pair)
with minimum distance\footnote{Both these problems are described here as search problems. However, our results also hold for their decision counterparts. In the decision versions, we are given additionally a real number $R>0$ as part of the input, and the goal is to determine if there exists a pair of points whose distance is at most $R$.}.
It is not hard to see that \bcp is at least as hard as \closest
since we can apply an algorithm for \bcp to solve \closest with a similar asymptotic running time.
However, it is not clear whether the other direction is true.

We will give a simple reduction from \bcp to \closest using a polar-pair of point-sets.
First, take a polar-pair $(A,B)$, each with cardinality $n=m/2$, in the $L^p$-metric.
Next, pair up vectors in $A$ and $B$ to the red and blue points (the inputs to \bcp), respectively,
and then concatenate a vector $u\in A$ (resp., $v\in B$) to its matching red (resp., blue) point.
This reduction increases the distances between every pair of points,
but by the definition of the polar-pair, this process has more effect on the distances of
the {\em monochromatic} (i.e., red-red or blue-blue) pairs
than that of {\em bichromatic pairs},
and the reduction, in fact, has no effect on the order of crossing-pair distances at all.
By scaling the vectors in $A$ and $B$ appropriately,
this gives an instance of \closest whose closest pair of points is bichromatic.
Consequently, provided that the polar-pair of point-sets $(A,B)$ in a $d$-dimensional metric can
be constructed within a running time at least as fast as the time 
for computing \closest in the same metric,
an algorithm for \closest can be used to solve \bcp in the same asymptotic running time. In other words, 
this gives a reduction from \bcp to \closest with an increase in dimension by $\bicd(L^p)$, and thus if we had a running time lower bound for \bcp, it implies the same running time lower bound for \closest when $d=\Omega(\bicd(L^p))$.


\subsubsection{Bounds on Contact Dimension and Sphericity of Biclique}\label{technical}

Our main technical results are lower and upper bounds on
the biclique contact-dimension for the $L^p$-metric space where $p\in \mathbb{R}_{\ge 1}\cup\{0\}$ ($L^0$ is the Hamming metric).

\begin{theorem}
\label{thm:main-all-dim-results}
The following are upper and lower bounds on
biclique contact-dimension for the $L^p$-metric.
\begin{align}
  \bisph(L^0) = \bicd(L^0) = n
\label{eq:L0-bounds}\\
n \leq \bisph(L^0_{\{0,1\}})) \leq \bicd(L^0_{\{0,1\}})
  \leq n^2
  && \mbox{(i.e., $P\subseteq \{0,1\}^d$)}
\label{eq:L0-bin-bounds}\\
\Omega(\log n) \leq \bisph(L^1) \leq \bicd(L^1) \leq n^2
\label{eq:L1-bounds}\\
\Omega(\log n) \leq \bisph(L^p) \leq \bicd(L^p) \leq 2n
  && \mbox{for $p\in (1,2)$}
\label{eq:L2-bounds}\\
  \bisph(L^p) = \Theta(\bicd(L^p)) = \Theta(\log n)
  && \mbox{for $p>2$} \label{eq:Lp-bounds}
\end{align}
\end{theorem}

Note that $\bisph(\Delta) \leq \bicd(\Delta)$ for any metric $\Delta$.
Thus, it suffices to prove a lower bound for $\bisph(\Delta)$ and
prove an upper bound for $\bicd(\Delta)$.

We note that the bounds on the sphericity and the contact dimension of the $L^1$-metric in (\ref{eq:L1-bounds}) are follow from (\ref{eq:Lp-bounds}) and (\ref{eq:L0-bin-bounds}), respectively.
We are unable to show a strong (e.g., linear) lower bound for the $L^1$-metric.
However, we prove the weaker (average-case) result below for the $L^1$-metric which can be seen as a progress toward proving stronger lower bounds on the sphericity of the biclique in this metric (see Corollary~\ref{lem:LB-point-set-l1-large-gap} for more discussion on its applications).

\begin{theorem}
\label{lem:NoDistL1}
For any integer $d>0$, there exist no two finitely supported random variables $X,Y$ taking values from $\reals^{d}$ such that the following holds.
\begin{align*}
\underset{x_{1},x_{2}\in_{R}X}{\mathbb{E}}\left[\left\Vert x_{1}-x_{2}\right\Vert _{1}\right]+
\underset{y_{1},y_{2}\in_{R}Y}{\mathbb{E}}\left[\left\Vert y_{1}-y_{2}\right\Vert _{1}\right]&>\underset{x_{1}\in_{R}X,\,y_{1}\in_{R}Y}{2\cdot \mathbb{E}}\left[\left\Vert x_{1}-y_{1}\right\Vert _{1}\right]\,.
\end{align*}
\end{theorem}

For an overview on the known bounds on $\bisph$ and $\bicd$ (including the results in this paper),
please see Table~\ref{tab:sphericity-contact-dim-bounds}. 

\begin{table}
\begin{center}
\begin{tabular}{ |p{2.5cm}||p{7.25cm}|p{3cm}|  }
\hline
\centering \textbf{Metric} & \centering  \textbf{Bound} &{\centering \ \ \ \ \ \ \  \textbf{From} }\\
 \hline
$L^0$   & $  \bisph(L^0) = \bicd(L^0) = n$  &  This paper\\
$L^1$&    $\Omega(\log n) \leq \bisph(L^1) \leq \bicd(L^1) \leq n^2$    &This paper\\
$L^p$, $p\in (1,2)$&   $\Omega(\log n) \leq \bisph(L^p) \leq \bicd(L^p) \leq 2n$    &  This paper\\
$L^2$    &   $ n< \bisph(L^2) \leq \bicd(L^2) <1.5 \cdot n$    & \cite{M91,FM88}\\
$L^p$, $p>2$ & $\bisph(L^{p})=\Theta(\bicd(L^{p}))=\Theta(\log{n})$ &This paper\\
$L^\infty$& $\bisph(L^{\infty})=\bicd(L^{\infty})=2\log_2{n}$ &\cite{R69cubicity}\\
 \hline
\end{tabular}
\end{center}
\caption{Known Bounds on Sphericity and Contact Dimension of Biclique}
\label{tab:sphericity-contact-dim-bounds}
\end{table}

In Appendix~\ref{sec:L2-dim}, we give an alternate proof of the linear lower  bound on $\bisph(L^2)$ using spectral analysis similar to that in \cite{BL05}. While our lower bound is slightly weaker than the best known bounds \cite{FM88,M91}, our arguments require no heavy machinery and thus are arguably simpler than the previous works \cite{FM88,M91,BL05}.

Alman and Williams \cite{AW15} showed the subquadratic-time hardness for \bcp in $L^p$-metrics, for all $p\in \mathbb{R}_{\ge 1}\cup\{0\}$,  under the {\em Orthogonal Vector Hypothesis} (OVH). From Theorem~\ref{thm:main-all-dim-results} and the connection between \bcp and \closest described in subsection~\ref{sec:connection}, we have the following hardness of \closest.

\begin{theorem}\label{thm:closestMainIntro}
Let $p>2$. For any $\varepsilon>0$ and $d=\omega(\log n)$,
the closest pair problem in
the $d$-dimensional $L^{p}$-metric admits
no $(n^{2-\varepsilon})$-time algorithm unless
the Orthogonal Vectors Hypothesis is false.
\end{theorem}

We remark here that showing conditional hardness for \closest in the $L^p$ metric for $p\le 2$ remains an outstanding open problem\footnote{The subquadratic-time hardness of \closest in the $L^p$-metric for $p\in \mathbb R_{\ge 1}\cup\{0\}$ was claimed in \cite{ARW17-FOCS17} but later retracted \cite{AR17}.}. Recently, Rubinstein \cite{R17} showed that the subquadratic-time hardness holds even for approximating \bcp:
Assuming OVH, for every $p\in \mathbb R_{\ge 1}\cup \{0\}$ and every $\varepsilon>0$, there is a constant $\gamma(\varepsilon,p)>0$ such that there is no $(1+\gamma)$-approximation algorithm running in time $O(n^{2-\varepsilon})$ for \bcp in the $L^p$-metric.
%
By using the connection between \bcp and \closest described in subsection~\ref{sec:connection} and the bounds in Theorem~\ref{thm:main-all-dim-results} (to be precise we need the efficient construction with appropriate gap as given by Theorem~\ref{thm:Codes}), the hardness of approximation result can be extended to \closest for $L^p$ metrics where $p>2$.

\begin{theorem}\label{thm:closestApproxIntro}
Let $p>2$. For every $\varepsilon>0$ and $d=\omega(\log n)$, there exists a constant $\gamma=\gamma(p,\varepsilon)>0$ such that
the closest pair problem in
the $d$-dimensional $L^{p}$-metric admits
no $(n^{2-\varepsilon})$-time $(1+\gamma)$-approximation algorithm unless
the Orthogonal Vectors Hypothesis is false.
\end{theorem}

We remark that the hardness for the case of the $L^{\infty}$-metric does not follow (at least directly) from \cite{AW15} or \cite{R17}.
For independent interest, we show the subquadratic-time
hardness of \bcp and \closest in the $L^{\infty}$-metric.

\begin{theorem}
\label{thm:closest-hardness-max-norm}
For any $\varepsilon>0$ and $d=\omega(\log n)$,
the closest pair problem in
the $d$-dimensional $L^{\infty}$-metric admits
no $(n^{2-\varepsilon})$-time $(2-o(1))$-approximation algorithm unless
the Orthogonal Vectors Hypothesis is false.  
\end{theorem}

We note that the lower bounds on $\bisph$ act as barriers for gadget reductions from \bcp to \closest. This partially explains why there has been no progress in showing conditional hardness for \closest in the Euclidean metric for $d=\omega(\log n)$ dimensions (as $\bisph(L^2)=\Omega(n)$). 
In addition, Rubinstein noted in \cite{R17} that one obstacle in proving inapproximability results for \closest is due to the triangle inequality -- any two point-sets $A$ and $B$ in \emph{any metric space} cannot have distinct points $a,a'\in A$ and $b\in B$ such that $\|a-a'\|> 2\cdot \max\{\|a-b\|,\|a'-b\|\}$ (as otherwise it would violate the triangle inequality). This rules out the possibility of obtaining the conditional hardness for $2$-approximating \closest for any metric via simple gadget reductions.  
We note that the inapproximability factor of Theorem~\ref{thm:closest-hardness-max-norm} matches the triangle inequality barrier (for the $L^\infty$ metric).

\subsection{Related Works}\label{sec:related}

While our paper studies sphericity and contact-dimension
of the complete bipartite graph,
determining the contact-dimension of a complete graph
in $L^p$-metrics has also been extensively studied in
the notion of {\em equilateral dimension}.
To be precise, the equilateral dimension of a metric $\Delta$
which is the maximum number of equidistant points that can be
packed in $\Delta$.
An interesting connection is in the case of the $L^1$-metric,
for which we are unable to establish a strong lower bound for
$\bisph(L^1)$.
The equilateral dimension of $L^1$ is known to be at least $2d$,
and this bound is believed to be tight~\cite{GUY83-KusnerConj}.
This is a notorious open problem known as {\em Kusner's conjecture}, which is confirmed for $d=2,3,4$
\cite{BCL98-eqdim,KLS00-eqdim-d4},
and the best upper bound for $d\geq 5$ is $O(d\log d)$ by Alon and Pudl\'ak
\cite{AP03-eqdim}.
If Kusner's conjecture is true for all $d$, then $\cd_1(K_n)=n/2$.


\begin{sloppypar}The complexity of \closest has been a subject of study for many decades.
There have been a series of developments on
\closest in the Euclidean space (see, e.g.,
\cite{Ben80,HNS88,KM95,SH75,BS76}),
which culminates in a deterministic $O(2^{O(d)}n\log n)$-time
algorithm \cite{BS76} and a randomized $O(2^{O(d)}n)$-time algorithm
\cite{Rabin76,KM95}.
For low (i.e., constant) dimensions, these algorithms are tight as
the matching lower bound of $\Omega(n\log n)$ was shown by Ben-Or
\cite{Ben83} and Yao \cite{Yao91} for the
{\em algebraic decision tree} model,
thus settling the complexity of \closest in low dimensions.
For high dimensions (i.e., $d=\omega(\log n)$), 
there is no known algorithm that runs in time significantly better
than a trivial $O(n^2d)$-time algorithm for general $d$
except for the case that $d \geq \Omega(n)$ whereas there are
subcubic-time algorithms in $L^1$ and $L^{\infty}$ metrics
\cite{GS16,ILLP04}.\end{sloppypar}

\begin{sloppypar}In the last few years, there has been a lot of progress in our understanding of \bcp, \closest, and related problems. Alman and Williams \cite{AW15} showed subquadratic time hardness for \bcp in $d=\omega(\log n)$ dimensions under OVH in the $L^p$ metric for every $p\in\mathbb R_{\ge 1}\cup\{0\}$. Williams \cite{W17diff} extended the result of \cite{AW15} and showed  
the above subquadratic-time hardness for \bcp even for dimensions $d=\omega((\log\log n)^2)$ under OVH. In a recent breakthrough on hardness of approximation in P,\ Abboud et al. \cite{AR17} showed the subquadratic-time hardness for approximating the Bichromatic Maximum Inner Product problem under OVH in the $L^p$ metric for every $p\in\mathbb R_{\ge 1}\cup\{0\}$, and the result holds for almost polynomial approximation factors.
More recently, building upon the ideas in \cite{AR17}, Rubinstein \cite{R17} showed under OVH the inapproximablility of \bcp for every $L^p$-metric for $p\in\mathbb R_{\ge 1}\cup \{0\}$.  \end{sloppypar}

%
%
%

%% file: Preliminaries.tex
\section{Preliminaries}
\label{sec:prelim}
We use the following standard terminologies and notations.

\paragraph{Metrics.}
\begin{sloppypar}
For any two vectors $a,b\in\mathbb{R}^d$, the distance between them in the $L^p$-metric is denoted by $||a-b||_p~=~\left(\sum_{i=1}^d|a_i-b_i|^p\right)^{1/p}$.
Their distance in the $L^{\infty}$-metric is denoted by 
$||a-b||_{\infty} = \underset{{i\in[d]}}{\max}\ \{|a_i-b_i|\}$, and
in the $L^0$-metric is denoted by 
$||a-b||_0=|\{a_i\neq b_i: i\in [d]\}|$, i.e.,
the number of coordinates on which $a$ and $b$ differ. 
The $L^p$-metrics that are well studied in literature are
the {\em Hamming metric} ($L^0$-metric),
the {\em rectilinear metric} ($L^1$-metric),
the {\em Euclidean metric} ($L^2$-metric), and
the {\em Chebyshev metric}  ($L^{\infty}$-metric). 
\end{sloppypar}

\paragraph{Problems.}
Here we give formal definitions of \closest and \bcp.
In \closest, we are given a collection of points $P\subseteq\reals^d$ in 
a $d$-dimensional $L^p$-metric, and the goal is to find a pair of distinct points $u,v\in P$
that minimizes $\|u-v\|_p$.
In \bcp, the input point-set is partitioned into two color classes 
(the collections of red and blue points) $A$ and $B$,
and the goal is to find a pair of points $u\in A$ and $v\in B$
that minimizes $\|u-v\|_p$.

\paragraph{Fine-Grained Complexity and Conditional Hardness.}
Conditional hardness is the current trend in proving 
running-time lower bounds for polynomial-time solvable problems.
This has now developed into the area of {\em Fine-Grained Complexity}.
Please see, e.g., \cite{Vir18,Williams15,Williams16} and references therein.

The {\em Orthogonal Vectors Hypothesis}  (OVH) is a popular complexity theoretic assumption in Fine-Grained Complexity. 
OVH states that in the Word RAM model with $O(\log n)$ bit words, 
any algorithm requires $n^{2-o(1)}$ time in expectation to determine whether 
collections of vectors $A,B\subseteq\{0,1\}^d$
with $|A|=|B|=n/2$ and $d = \omega(\log n)$
contain an orthogonal pair $u\in A$ and $v\in B$ (i.e.,  $\sum_{i=1}^du_i\cdot v_i=0$).  We emphasize that the scalar product is taken over the field of real numbers and not modulo 2.

Another popular conjecture is 
the {\em Strong Exponential-Time Hypothesis} for SAT (SETH), 
which states that, for every $\varepsilon > 0$, 
there exists an integer $k_{\varepsilon}$ such that $k_{\varepsilon}$-SAT
on $n$ variables cannot be solved in 
$O(2^{(1-\varepsilon)n})$-time.
It was shown by Williams that SETH implies OVH \cite{W05}.

%% file: ell1.tex
\section{Geometric Representation of Biclique in $L^1$}
\label{sec:L1-possibility}

In this section, we discuss the case of the $L^1$-metric.
As discussed in the introduction, this is the only case where
we are unable to prove neither strong lower bound nor linear upper bound.
A weak lower bound of $\bisph(L^1)\geq\Omega(\log n)$ follows from
the proof for the $L^p$-metric with $p\ge 1$ in 
Section~\ref{sec:Lp-dim:lower-bound} (Theorem~\ref{thm:lower-bound-Lp}),
and a quadratic upper bound $\bicd(L^1)\leq n^2$ follows from 
the proof for the $L^0$-metric 
in Section~\ref{sec:L0-dim:upper-bound} (Corollary~\ref{cor:ub-point-set-l0-binary}).
However, we cannot prove any upper bound smaller than $\Omega(n^2)$ or
any lower bound larger than   $O(\log n)$.
Hence, we study an average case relaxation of the question.

We show in Theorem~\ref{lem:NoDistL1} that there is no distribution
whose expected distances simulate a polar-pair of point-sets in the
$L^1$-metric. 
Consequently, even though we could not prove the biclique sphericity
lower bound for the $L^1$-metric, we are able to refute an existence
of a geometric representation with {\em large gap} for any dimension 
as shown in Corollary~\ref{lem:LB-point-set-l1-large-gap}.
(A similar result was shown in \cite{DM94} for the $L^2$-metric.)

  \renewcommand{\sup}{\mathsf{supp}}
\begin{definition}[$L^1$-distribution]
\label{def:ell1-distribution} For any $d>0$, let $X,Y$
be two random variables taking values from $\reals^{d}$. An $L^{1}$-distribution
is constructed by $X,Y$ if the following holds.
\begin{align}
\underset{x_{1},x_{2}\in_{R}X}{\mathbb{E}}\left[\left\Vert x_{1}-x_{2}\right\Vert _{1}\right]+\underset{y_{1},y_{2}\in_{R}Y}{\mathbb{E}}\left[\left\Vert y_{1}-y_{2}\right\Vert _{1}\right]>\underset{x_{1}\in_{R}X,\,y_{1}\in_{R}Y}{2\cdot \mathbb{E}}\left[\left\Vert x_{1}-y_{1}\right\Vert _{1}\right]\,.   \label{eq3}
\end{align}
\end{definition}

{
\renewcommand{\thetheorem}{\ref{lem:NoDistL1}}
\begin{theorem}[Restated] 
For any two finitely supported random variables $X,Y$ that are taking
values from $\reals^{d}$, there is no $L^{1}${-distribution.}
\end{theorem}
\addtocounter{theorem}{-1}
}
\begin{proof}
Assume towards a contradiction that there exist two finitely supported
random variables $X,Y$ that take values in $\reals^{d}$
and satisfy Eq.~\ref{eq3} of Definition~\ref{def:ell1-distribution}.
Given a vector $x\in\reals^{d}$, we denote by $x\left(i\right)$
the value of the $i$-th coordinate of $x$. Hence, the following inequalities
hold:
\begin{align*}
0 & >2\cdot\underset{\substack{x_{1}\in_{R}X\\ y_{1}\in_{R}Y}}{\mathbb{E}}\left[\left\Vert x_{1}-y_{1}\right\Vert _{1}\right]-\underset{x_{1},x_{2}\in_{R}X}{\mathbb{E}}\left[\left\Vert x_{1}-x_{2}\right\Vert _{1}\right]-\underset{y_{1},y_{2}\in_{R}Y}{\mathbb{E}}\left[\left\Vert y_{1}-y_{2}\right\Vert _{1}\right]\\
 & =2\cdot\underset{\substack{x_{1}\in_{R}X\\ y_{1}\in_{R}Y}}{\mathbb{E}}\left[d\cdot \underset{i\in_{R}\left[d\right]}{\mathbb{E}}\left[|x_1(i)-y_1(i)|\right]\right]-\underset{x_{1},x_{2}\in_{R}X}{\mathbb{E}}\left[d\cdot \underset{i\in_{R}\left[d\right]}{\mathbb{E}}\left[|x_1(i)-x_2(i)|\right]\right]\\
 &\phantom{=2\cdot\underset{\substack{x_{1}\in_{R}X\\ y_{1}\in_{R}Y}}{\mathbb{E}}\left[d\cdot \underset{i\in_{R}\left[d\right]}{\mathbb{E}}\left[|x_1(i)-y_1(i)|\right]\right]\ }-\underset{y_{1},y_{2}\in_{R}Y}{\mathbb{E}}\left[d\cdot \underset{i\in_{R}\left[d\right]}{\mathbb{E}}\left[|y_1(i)-y_2(i)|\right]\right]\\
 & =d\cdot \underset{i\in_{R}\left[d\right]}{\mathbb{E}}\left[2\cdot\underset{\substack{x_{1}\in_{R}X\\ y_{1}\in_{R}Y}}{\mathbb{E}}\left[|x_1(i)-y_1(i)|\right]-\underset{x_{1},x_{2}\in_{R}X}{\mathbb{E}}\left[|x_1(i)-x_2(i)|\right]-\underset{y_{1},y_{2}\in_{R}Y}{\mathbb{E}}\left[|y_1(i)-y_2(i)|\right] \right]
 \,.
\end{align*}
Thus, for some $i^{\star}\in[d]$, the following holds:
\begin{align}
0>2\cdot\underset{\substack{x_{1}\in_{R}X\\ y_{1}\in_{R}Y}}{\mathbb{E}}\left[|x_1(i^*)-y_1(i^*)|\right]-\underset{x_{1},x_{2}\in_{R}X}{\mathbb{E}}\left[|x_1(i^*)-x_2(i^*)|\right]-\underset{y_{1},y_{2}\in_{R}Y}{\mathbb{E}}\left[|y_1(i^*)-y_2(i^*)|\right]\,. \label{reduce}
\end{align}

Fix $i^\star\in [d]$ satisfying the above inequality. 
For the sake of clarity, we assume that the random variables $X,Y$ are taking values in $\reals$ (i.e., projection on the $i^{\star\text{th}}$ coordinate).
We can assume that the size of $\sup\left(X\right)\cup \sup\left(Y\right)$ is greater than 1 because if $\sup\left(X\right)\cup \sup\left(Y\right)$ contains a single point,
then 
$$\underset{x_{1}\in_{R}X,\,y_{1}\in_{R}Y}{\mathbb{E}}\left[\left\Vert x_{1}-y_{1}\right\Vert _{1}\right]=\underset{x_{1},x_{2}\in_{R}X}{\mathbb{E}}\left[\left\Vert x_{1}-x_{2}\right\Vert _{1}\right]=\underset{y_{1},y_{2}\in_{R}Y}{\mathbb{E}}\left[\left\Vert y_{1}-y_{2}\right\Vert _{1}\right]=0,$$ 
contradicting Eq.~\ref{reduce}. 
Let $\sup\left(X\right)\cup \sup\left(Y\right)$ contain $t\ge 2$ points. We prove by induction on $t$, that there are no $X,Y$ over $\reals$ satisfying Eq.~\ref{reduce}. 

The base case is when $t=2$.  Let $\sup\left(X\right)\cup \sup\left(Y\right)=\{a,b\}$ and $p,q\in [0,1]$ be such that $\Pr[a\in X]=p$ and $\Pr[a\in Y]=q$. We have $\underset{x_{1},x_{2}\in_{R}X}{\mathbb{E}}\left[|x_1-x_2|\right]=2p(1-p)|a-b|$ and $\underset{y_{1},y_{2}\in_{R}Y}{\mathbb{E}}\left[|y_1-y_2|\right]=2q(1-q)|a-b|$, and $\underset{\substack{x_{1}\in_{R}X\\ y_{1}\in_{R}Y}}{\mathbb{E}}\left[|x_1-y_1|\right]=\left(p(1-q)+q(1-p)\right)|a-b|$. Substituting the above in Eq.~\ref{reduce}, we have:
$$
0>2\left(q(1-p)+p(1-q)\right) -2p(1-p)-2q(1-q) = 2\cdot (p-q)^2,
$$
a contradiction.

Assume the induction hypothesis that there are no $X,Y$ taking values from $\reals$ satisfying Eq.~\ref{reduce} when the size of $\sup\left(X\right)\cup \sup\left(Y\right)$  is equal to $k\ge 2$.
Then consider the case when $t=k+1\ge 3$.
Sort the points in $\sup\left(X\right)\cup \sup\left(Y\right)$ by their values, and denote by $s_{i}$ the value of the $i$-th point of $\sup\left(X\right)\cup \sup\left(Y\right)$.
For the sake of simplicity, we say that we \emph{change} the value of $s_{t-1}$ to $\tilde{s}_{t-1}$, where $s_{t-2}\le\tilde{s}_{t-1}\le s_{t}$, if after changing its value
we change the values of (at least one of) $X,Y$ to $\tilde{X},\tilde{Y}$
in such a way that the value of the $(t-1)$-th point (after sorting)
of $\sup\left(\tilde{X}\right)\cup \sup\left(\tilde{Y}\right)$ is equal
to $\tilde{s}_{t-1}$ (if $s_{t-2}=\tilde{s}_{t-1}$, then the value of
the $\left(t-2\right)$-th point of $\sup\left(\tilde{X}\right)\cup \sup\left(\tilde{Y}\right)$
is equal to $\tilde{s}_{t-1}$). Define the
function $f:\,\left[s_{t-2},s_{t}\right]\rightarrow\reals$
as follows:  $$f\left(x\right)=2\cdot\underset{\substack{x_{1}\in_{R}\tilde X\\ y_{1}\in_{R}\tilde Y}}{\mathbb{E}}\left[|x_1-y_1|\right]-\underset{x_{1},x_{2}\in_{R}\tilde X}{\mathbb{E}}\left[|x_1-x_2|\right]-\underset{y_{1},y_{2}\in_{R}\tilde Y}{\mathbb{E}}\left[|y_1-y_2|\right]\,,
$$
where $\tilde{X},\tilde{Y}$ are obtained after changing $s_{t-1}$
to $x\in\left[s_{t-2},s_{t}\right]$.
The crucial observation is that the function $f$ is linear.
Hence, either $f\left(s_{t-2}\right)\ge f\left(s_{t-1}\right)$
or $f\left(s_{t}\right)\ge f\left(s_{t-1}\right)$, and we can reduce the size of $\sup\left(X\right)\cup \sup\left(Y\right)$ by 1.
However, this contradicts our induction hypothesis.
\end{proof}

The following corollary refutes the existence of a polar-pair of point-sets
with large gap in any dimension.

\begin{corollary}[No Polar-Pair of Point-Sets in $L^1$ with Large Gap]
\label{lem:LB-point-set-l1-large-gap}
For any $\alpha>0$,
there exist no subsets $A,B\subseteq \reals^d$ of $n/2$ vectors
with $d<n/2$ such that
\begin{itemize}
\item For any distinct $u,v$ both in $A$, or both in $B$, $\|u-v\|_1 \ge \frac{1}{1-2/n} \cdot \alpha$.
\item For any $u\in A$ and $v\in B$, $\|u-v\|_1 < \alpha$.
\end{itemize}
\end{corollary}
\begin{proof}
Assume towards a contradiction that there exist a polar-pair of point-sets $(A,B)$ in the $L^1$-metric that satisfies the conditions above.
We can create a distribution $X$ and $Y$ such that
\[
\underset{x_{1},x_{2}\in_{R}X}{\mathbb{E}}\left[\left\Vert x_{1}-x_{2}\right\Vert _{1}\right]
= \underset{y_{1},y_{2}\in_{R}Y}{\mathbb{E}}\left[\left\Vert y_{1}-y_{2}\right\Vert _{1}\right]
>
\underset{x\in_{R}X,y\in_{R}Y}{\mathbb{E}}\left[\left\Vert x-y\right\Vert _{1}\right].
\]
To see this, we create a uniform random variable $X$ (resp., $Y$) for the set $A$ (resp., $B$).
Now the expected distance of two independent copies of $X$ (resp., $Y$) is at least $\frac{1}{1-2/n}\alpha\cdot \left(1-\frac{1}{n/2}\right) = \alpha$,
which follows because we may pick the same point twice.
Since the expected distance of the crossing pair $u\in A$ and $v\in B$ is less than $\alpha$.
This contradicts Theorem~\ref{lem:NoDistL1}.
\end{proof}

We can show similar results that there are no polar-pairs of point-sets with large gap in the $L^0$ and $L^2$ metrics.
The case of the $L^0$-metric follows directly from Theorem~\ref{lem:NoDistL1} when the alphabet set is $\{0,1\}$.
(Please also see Lemma~\ref{lem:no-distribution-for-l0} for an alternate proof.)
The case of the $L^2$-metric follows from the fact that $\bisph(L^2) = \Omega(n)$ \cite{FM88,M91}
and that we can reduce the dimension of  polar-pairs of point-sets with constant gap to $O(\log n)$
using dimension reduction \cite{JL84}.

%% file: ell0.tex
\section{Geometric Representation of Biclique in $L^0$}
\label{sec:L0-dim}

In this section, we prove a lower bound 
on $\bisph(L^0)$ and  an upper bound on $\bicd(L^0)$.
We start by providing a real-to-binary reduction below.
Then we proceed to prove the lower bound on $\bisph(L^0)$
in Section~\ref{sec:L0-dim:lower-bound}
and then the upper bounds on $\bicd(L^0)$ in
Section~\ref{sec:L0-dim:upper-bound}.

\paragraph*{Real to Binary Reduction.}
\label{sec:L0-dim:str-to-bin}

First we prove the following (trivial) lemma, which allows mapping
from vectors in $\reals^d$ to zero-one vectors.

\begin{lemma}[Real to Binary Reduction]
\label{lem:trivial-real-to-binary}
Let $S\subseteq \reals$ be a finite set of real numbers.
Then there exists a transformation
$\phi:S^d\rightarrow\{0,1\}^{d|S|}$
such that, for any $x,y\in S^d$,
\[
\|x-y\|_0 = \frac{1}{2}\cdot\|\phi(x)-\phi(y)\|_0
\]
\end{lemma}
\begin{proof}
First, we order the elements in $S$ in an arbitrary order
and write it as $S=\{r_1,r_2,\ldots,r_{|S|}\}$.
Next we define $\psi:S\rightarrow\{0,1\}^{|S|}$ so that the $i^{\text{th}}$
coordinate of $\psi(r_i)$ is $1$, and the rest are zeroes. 
That is, 
\[
\psi(r_i)_j=\left\{
\begin{array}{ll}
1 & \mbox{if $j=i$}\\
0 & \mbox{otherwise}
\end{array}\right.
\]

Then we define $\phi(x)=(\psi(x_1),\psi(x_2),\ldots,\psi(x_d))$. 
Clearly, $\|\psi(r_i)-\psi(r_j)\|_0=2$ if and only if $r_i\neq r_j$.
Therefore, we conclude that for any $x,y\in S^d$,
\[
\|\phi(x)-\phi(y)\|_0 = 2\cdot \|x-y\|_0.\qedhere
\]
\end{proof}

\subsection{Lower Bound on the Biclique-Sphericity}
\label{sec:L0-dim:lower-bound}

Now we will show that $\bisph(L^0) \geq n$.
Our proof requires the following lemma, which rules out 
a randomized algorithm that generates a polar-pair of point-sets.

\begin{lemma}[No Distribution for $L^0$]
\label{lem:no-distribution-for-l0}
For any $\alpha > \beta \geq 0$,
regardless of dimension,
there exist no distributions $\mathcal{A}$ and $\mathcal{B}$ 
of points in $\reals^d$ with finite supports such that 
\begin{itemize}
\item $\underset{x,x'\in_{R}\mathcal{A}}{\mathbb{E}}[\|x-x'\|_0] \geq \alpha$.
\item $\underset{y,y'\in_{R}\mathcal{B}}{\mathbb{E}}[\|y-y'\|_0] \geq \alpha$.
\item $\underset{\substack{x\in_{R}\mathcal{A}\\ y\in_{R}\mathcal{B}}}{\mathbb{E}}[\|x-y\|_0] \leq \beta$.
\end{itemize}
\end{lemma}
\begin{proof}
We prove the lemma by contradiction.
Assume to the contrary that such distributions exist.
Then 
\begin{equation}
\label{eq:L0-expected-distance}
\underset{x,x'\in_{R}\mathcal{A}}{\mathbb{E}}[\|x-x'\|_0] +
\underset{y,y'\in_{R}\mathcal{B}}{\mathbb{E}}[\|y-y'\|_0]
-
2\cdot \underset{\substack{x\in_{R}\mathcal{A}\\ y\in_{R}\mathcal{B}}}{\mathbb{E}}[\|x-y\|_0]
> 0.
\end{equation}

Let $A$ and $B$ be supports of $\mathcal{A}$ and $\mathcal{B}$,
respectively. 
By Lemma~\ref{lem:trivial-real-to-binary},
we may assume that vectors in $A$ and $B$ are binary vectors.  
Observe that each coordinate of vectors in $A$ and $B$ 
contribute to the expectations independently.
In particular, Eq.~\eqref{eq:L0-expected-distance}
can be written as 
\begin{equation}
\label{eq:L0-expand-the-sum}
2\sum_{i}\rho^A_{0,i}\rho^A_{1,i}
+ 2\sum_{i}\rho^B_{0,i}\rho^B_{1,i}
- 2\sum_{i}\left(\rho^A_{0,i}\rho^B_{1,i} 
+ \rho^B_{0,i}\rho^A_{1,i}\right)
>0
\end{equation}
where $\rho^A_{0,i}$, $\rho^A_{1,i}$, $\rho^B_{0,i}$ and
$\rho^B_{1,i}$ are the probability that 
the $i^{\text{th}}$ coordinate of $x\in A$ (resp., $y\in B$) is $0$ 
(resp., $1$). 
Thus, to show a contradiction, it is sufficient to consider the
coordinate which contributes the most to the summation in 
Eq.~\eqref{eq:L0-expand-the-sum}.
The contribution of this coordinate to the summation is 
\begin{equation}
\label{eq:L0-one-coordinate-contribution}
2\rho^A_0\rho^A_1 + 2\rho^B_0\rho^B_1 - 
2(\rho^A_0\rho^B_1 + \rho^A_1\rho^B_0)
= 2(\rho^A_0(\rho^A_1-\rho^B_1) +
  2(\rho^B_0(\rho^B_1-\rho^A_1)
= 2(\rho^A_0-\rho^B_0)(\rho^A_1-\rho^B_1)
\end{equation}
Since $\rho^A_0+\rho^A_1=1$ and 
      $\rho^B_0+\rho^B_1=1$,
the summation in Eq.\eqref{eq:L0-one-coordinate-contribution}
can be non-negative only if 
$\rho^A_0=\rho^B_0$ and $\rho^A_1=\rho^B_1$.
But, then this implies that 
the summation in Eq.\eqref{eq:L0-one-coordinate-contribution} is zero.
We have a contradiction since this coordinate contributes
the most to the summation in Eq.~\eqref{eq:L0-expand-the-sum}
which we assume to be positive.
\end{proof}

The next Theorem shows that $\bisph(L^0) \geq n$.

\begin{theorem}[Lower Bound for $L^0$ with Arbitrary Alphabet]
\label{thm:LB-point-set-l0-arbitrary}
For any integers $\alpha>\beta\geq 0$ and $n>0$,
there exist no subsets $A,B\subseteq \reals^d$ of $n$ vectors
with $d<n$ such that 
\begin{itemize}
\item For any distinct $a,a'\in A$, $\|a-a'\|_0 \geq \alpha$.
\item For any distinct $b,b'\in B$, $\|b-b'\|_0 \geq \alpha$.
\item For any $a\in A$ and $b\in B$, $\|a-b\|_0 \le\beta$.
\end{itemize}
\end{theorem}
\begin{proof}
Suppose for a contradiction that such subsets $A$ and $B$ exist
with $d<n$. 
We build uniform distributions $\mathcal{A}$ and $\mathcal{B}$
by uniformly at random picking a vector in $A$ and $B$, respectively.
Then it is easy to see that the expected value of inner distance
is 
\[
 \underset{x,x'\in_{R}\mathcal{A}}{\mathbb{E}}[\|x-x'\|_0] \geq \alpha-\frac{\alpha}{n}
\]
The inner distance of $B$ is similar.
We know that $\alpha-\beta\geq 1$ because they are integers
and so are $L^0$-distances.
But, then if $\alpha < n$, we would have distributions
that contradict Lemma~\ref{lem:no-distribution-for-l0}.
Note that $\alpha$ and $\beta$ are at most $d$ (dimension).
Therefore, we conclude that $d\geq n$.
\end{proof}


\subsection{Upper Bound on the Biclique Contact-Dimension}
\label{sec:L0-dim:upper-bound}

Now we show that 
$\bicd(L^0)\leq n$.

\begin{theorem}[Upper Bound for $L^0$ with Arbitrary Alphabet]
\label{thm:ub-point-set-l0-arbitrary}
For any integer $n>0$ and $d=n$, 
there exist subsets $A,B\subseteq \reals^d$ each with 
$n$ vectors such that
\begin{itemize}
\item For any distinct $a,a'\in A$, $\|a-a'\|_0 = d$.
\item For any distinct $b,b'\in B$, $\|b-b'\|_0 = d$.
\item For any $a\in A$ and $b\in B$, $\|a-b\|_0 = d-1$.
\end{itemize}
\end{theorem}
\begin{proof}

First we construct a set of vectors $A$.
For $i=1,2,\ldots,n$, we define the $i^{\text{th}}$ vector $a$ of $A$
so that $a$ is an all-$i$ vector. That is,
\[
a=(i,i,\ldots,i).
\]

Next we construct a set of vectors $B$.
The first vector of $B$ is $(1,2,\ldots,n)$.
Then the $(i+1)^{\text{th}}$ vector of $B$ is the left rotation of the $i^{\text{th}}$
vector. 
Thus, the $i^{\text{th}}$ vector of $B$ is
\[
b=(i,i+1,\ldots,n,1,2,\ldots,i-1).
\]

It can be seen that the $L^0$-distance between any two vectors from the
same set is $d$ because all the coordinates are different.
Any vectors from different set, say $a\in A$ and $b\in B$,
must have at least one common coordinate.
Thus, their $L^0$-distance is $d-1$.
This proves the lemma.
\end{proof}

Below is the upper bound for zero-one vectors,
which is a corollary of Theorem~\ref{thm:ub-point-set-l0-arbitrary}.

\begin{corollary}[Upper Bound for $L^0$ with Binary Vectors]
\label{cor:ub-point-set-l0-binary}
For any integer $n>0$ and $d=n^2$,
there exist subsets $A,B\subseteq \reals^d$ each with $n$ vectors
such that
\begin{itemize}
\item For any distinct $a,a'\in A$, $\|a-a'\|_0 = n$.
\item For any distinct $b,b'\in B$, $\|b-b'\|_0 = n$.
\item For any $a\in A$ and $b\in B$, $\|a-b\|_0 = n-1$.
\end{itemize}
\end{corollary}
\begin{proof}
We take the construction from
Theorem~\ref{thm:ub-point-set-l0-arbitrary}.
Denote the two sets by $A'$ and $B'$,
and denote their dimensions by $d'=n$.

We transform $A'$ and $B'$ to sets $A$ and $B$
by applying the transformation $\phi$ in
Lemma~\ref{lem:trivial-real-to-binary}.
That is,
\[
A=\{\phi(a):a\in A'\} \quad\mbox{and}\quad
B=\{\phi(b):b\in B'\}.
\]
Since the alphabet set in Lemma~\ref{lem:trivial-real-to-binary}
is $[n]$, we have a construction of $A$ and $B$
with dimension $d=n^2$.
\end{proof}


%% file: ellp_1to2.tex
\section{Geometric Representation of Biclique in $L^p$ for $p\in (1,2)$}
\label{sec:Lp-between-1-and-2}

In this section, we prove the upper bound on $\bicd(L^p)$ for $p\in(1,2)$.
We are unable to show any lower bound for these $L^p$-metrics except
for the lower bound of $\Omega(\log n)$ obtained from
the $\epsilon$-net lower bound in Theorem~\ref{thm:lower-bound-Lp} 
(which will be proven in the next Section).

\begin{theorem}[Upper Bound for $L^p$ with $1 < p < 2$]
\label{thm:ub-point-set-lp12}
For every $1 < p < 2$ and for all integers $n\geq 1$,
there exist two sets $A,B \subseteq \reals^{2n}$
each of cardinality $n$ such that the following holds for some $s<2^{1/p}$:
\begin{enumerate}
\item For every distinct points $u,v\in A$, $\|u-v\|_p=2^{1/p}$.
\item For every distinct points $u,v \in B$, $\|u-v\|_p=2^{1/p}$.
\item For every pair of points $u\in A$ and $v\in B$, $\|u-v\|_p=s$.
\end{enumerate}
\end{theorem}
\begin{proof}
We will construct point-sets as claimed in the theorem for given $p$ and $n$.
Let $\alpha$ be a parameter depending on $p$ and $n$,
which will be set later.
For each $i\in [n]$, we create a point $a\in A$ by setting
\[
a_j = \left\{
  \begin{array}{ll}
     0 & \mbox{if $1 \leq j \leq n$ and $j \neq i$}\\
     1 & \mbox{if $1 \leq j \leq n$ and $i = j$}\\
     \alpha & \mbox{if $n+1 \leq j \leq 2n$}
  \end{array}\right.
\]
Similarly, for each $i\in [n]$, we create a point $b\in B$ by setting
\[
b_j = \left\{
  \begin{array}{ll}
     \alpha & \mbox{if $1 \leq j \leq n$}\\
     0 & \mbox{if $n+1 \leq j \leq 2n$ and $j\neq n+i$}\\
     1 & \mbox{if $n+1 \leq j \leq 2n$ and $j = n+i$}
  \end{array}\right.
\]

By construction, for every pair of distinct points $u,v$ both in $A$ or both in $B$,
their $L^p$-distance is  $\|u-v\|_p=2^{1/p}$,
and for every pair of points from different sets, say $u\in A$ and $v\in B$,
their $L^p$-distance is 
\begin{align}
\|u-v\|_p = 2^{1/p}\cdot ((1-\alpha)^p+(n-1)\cdot \alpha^p)^{1/p} \le 2^{1/p}\cdot ((1-\alpha)^p+n\cdot \alpha^p)^{1/p} 
\label{eq:Lp-12-diff-set}
\end{align}

Now we show that when $\alpha =\frac{1}{1+n^{\frac{1}{p-1}}}$ we have  $(1-\alpha)^p+n\cdot\alpha^p<1$ and thus the right hand side in Eq.~\eqref{eq:Lp-12-diff-set} is less than $2^{1/p}$ and the theorem follows.
Define a function $f(x) = (1-x)^p+ n\cdot x^p$. Note that $f'(x)=-p(1-x)^{p-1}+npx^{p-1}$. We have $f(0)=1$ and $f'(0)=-p$. Since $p>1$, we have that $f'$ is continuous and $f'(x)=0$ for $x \in [0,\alpha]$ if any only if $x=\alpha$. Thus, we have $f(\alpha)<f(0)=1$. 
\end{proof}

We note that the above upper bound holds for all $L^p$-metrics when $p>1$. It is just that for $p>2$ we have a better upper bound on $\bicd(L^p)$ (see Theorem~\ref{thm:Codes}).

%% file: ellp.tex
\section{Geometric Representation of Biclique in $L^p$ for $p>2$}
\label{sec:Lp-dim}

In this section, we show the lower bound on $\bisph(L^p)$
and an upper bound on $\bicd(L^p)$ for $p>2$.
Both bounds are logarithmic.
The latter upper bound is constructive and efficient
(in the sense that the polar-pair of point-sets can be constructed
in $\widetilde O(n)$-time).
This implies the subquadratic-time equivalence between \closest and \bcp.


\subsection{Lower Bound on the Biclique Sphericity}
\label{sec:Lp-dim:lower-bound}


Now we show the lower bound on the biclique sphericity 
of a complete bipartite graph in $L^p$-metrics with $p>2$.
In fact, we prove the lower bound for the case of
a star graph on $n$ vertices, denoted by $S_{n}$, and 
then use the fact that $\bisph(H) \leq \bisph(G)$ for all induced
subgraph $H$ of $G$ (i.e., 
$\bisph(K_{n/2,n/2},L^p) \geq \bisph(S_{n/2},L^p)$).

In short, we show in Theorem~\ref{thm:lower-bound-Lp}
that in $\mathbb{R}^d$, $2^{O(d)}$ is the maximum number of $L^p$-balls of radius
$1/2$ that we can pack in an $L^p$-ball of radius $3/2$ so that no
two of them intersect or touch each other.
This upper bound, in turn, implies the lower bound on the dimension.
We proceed with the proof by volume arguments, which are commonly used in
proving the minimum number of points in an {\em $\epsilon$-net}
that are sufficient to cover all the points in a sphere.

\begin{definition}[$\epsilon$-net]
The unit $L^{p}$-ball in $\reals^{d}$ centered at $o$ is denoted by
$$\mathbb{B}\left(L_{p}^{d},o\right)=\left\{ x\in\reals^{d}\,|\,\left\Vert
x-o\right\Vert _{p}\le1\right\}.$$
For brevity, we write $\mathbb{B}\left(L_{p}^{d}\right)$ to mean $\mathbb{B}\left(L_{p}^{d},\vec{0}\right)$.
Let $\left(X,d\right)$ be a metric space and let $S$
be a subset of $X$ and $\epsilon$ be a constant greater than $0$.
A subset $N_{\epsilon}$ of $X$ is called an {\bf $\epsilon$-net} of $S$
under $d$ if for every point $x\in S$ it holds for some point $y\in
N_{\epsilon}$
that $d\left(x,y\right)\le\epsilon$.
\end{definition}

The following lemma is well known in literature (see, e.g., \cite{vershynin2010}).
For the sake of completeness, we provide a proof below.

\begin{lemma}
\label{lem:nets}
There exists an $\epsilon$-net for $\mathbb{B}\left(L_{p}^{d}\right)$
under the $L^{p}$-metric of cardinality
$\left(1+\frac{2}{\epsilon}\right)^{d}$.
\end{lemma}
\begin{proof}
Let us fix $\epsilon>0$ and choose $N_{\epsilon}$ of maximal
cardinality (i.e., maximal under inclusion) such that $\left\Vert x-y\right\Vert _{p}>\epsilon$ for
all $x\neq y$ both in $N_{\epsilon}$.
We claim that $N_{\epsilon}$ is an $\epsilon$-net of the
$\mathbb{B}\left(L_{p}^{d}\right)$.
Otherwise, there would exist a point $x\in \mathbb{B}\left(L_{p}^{d}\right)$ that
is at least $\epsilon$-far from all points in $N_{\epsilon}$. Thus,
$N_{\epsilon}\cup\left\{ x\right\} $ contradicts the  maximality of
$N_{\epsilon}$. 
After establishing that $N_{\epsilon}$ is an $\epsilon$-net, we note
that by the triangle inequality, we have that the balls of radii
$\epsilon/2$ centered at the points in $N_{\epsilon}$ are disjoint.
On the other hand, by the triangle inequality all such balls lie in
$\left(1+\epsilon/2\right) \mathbb{B}\left(L_{p}^{d}\right)$.
Comparing the volumes gives us that
$\text{vol}\left(\frac{\epsilon}{2} \mathbb{B}\left(L_{p}^{d}\right)\right)\cdot
\left|N_{\epsilon}\right|\le
\text{vol}\left(\left(1+\frac{\epsilon}{2}\right) \mathbb{B}\left(L_{p}^{d}\right)\right)$.
Since $\text{vol}\left(r\cdot
\mathbb{B}\left(L_{p}^{d}\right)\right)=r^{d}\cdot
\text{vol}\left(\mathbb{B}\left(L_{p}^{d}\right)\right)$ for all $r \ge 0$, we conclude that 
$\left|N_{\epsilon}\right|\le\frac{\left(1+\epsilon/2\right)^{d}}{\left(\epsilon/2\right)^{d}}=\left(1+\frac{2}{\epsilon}\right)^{d}$ .
\end{proof}

\begin{theorem}
\label{thm:lower-bound-Lp}
For every $N,d\in \mathbb{N}$, for $p\geq 1$, and for any two sets $A,B\subseteq\reals^{d}$,
each of cardinality $N$, suppose the following holds
for some non-negative real numbers $\alpha$ and $\beta$ with $\alpha>\beta$.
\begin{enumerate}
\item For every distinct $u$ and $v$ both in $A$, $\|u-v\|_p>\alpha$.
\item For every distinct $u$ and $v$ both in $B$, $\|u-v\|_p>\alpha$.
\item For  every $u$ in $A$ and $v$ in $B$, $\|u-v\|_p\le \beta$.
\end{enumerate}
Then the dimension $d$ must be at least $\log_{5}(N)$.

\end{theorem}

\begin{proof}
Scale and translate the sets $A,B$ in such a way  that  $\beta=1$  and  that $\vec{0}\in B$.
It follows that $A\subseteq \mathbb{B}\left(L_{p}^{d}\right)$. By Lemma~\ref{lem:nets}, we can fix a $1/2$-net $N_{1/2}$ for $\mathbb{B}\left(L_{p}^{d}\right)$ of size  $5^{d}$. Note that, for every $x\in   N_{1/2}$, the ball $1/2 \cdot \mathbb{B}\left(L_{p}^{d},x\right)$ contains at most one point from $A$. Note also that $N_{1/2}$ covers $\mathbb{B}\left(L_{p}^{d}\right)$. Thus, $\left|A\right|\le 5^d$ which implies that $d\ge \log_{5}(N)$.
\end{proof}


\subsection{Upper Bound on the Biclique Contact-Dimension}
\label{sec:Lp-dim:upper-bound}

We first give a simple randomized construction that gives
a logarithmic upper bound on the biclique contact-dimension of $L^p$.
The construction is simple.
We uniformly at random take a subset $A$ of $n$ vectors from 
$\{-1,1\}^{d/2}\times \{0\} ^{d/2}$
and a subset $B$ of $n$ vectors from 
$\{0\}^{d/2}\times \{-1,1\}^{d/2}$.
Observe that, for any $p>2$, the $L^p$-distance of any pair of
vectors $u\in A$ and $v\in B$ is exactly $d^{1/p}$ while
the {\em expected distance} between the inner pair $u,u'\in A$ 
(resp., $v,v'\in B$) is strictly larger than $d^{1/p}$.
Thus, if we choose $d$ to be sufficiently large, e.g., 
$d\geq 10\ln n$,
then we can show by a standard concentration bound 
(e.g., Chernoff's bound) that the probability that
the inner-pair distance is strictly larger than $d$ is 
at least $1-1/n^{3}$. 
Applying the union bound over all inner-pairs, we have that
the $d^{1/p}$-neighborhood graph of $A\cup B$ is a bipartite complete graph 
with high probability.
Moreover, the distances between any crossing pairs $u\in A$ and $v\in B$
are the same for all pairs.
This shows the upper bound for the contact-dimension of a biclique in
the $L^p$-metric for $p>2$. 

The above gives a simple proof of the upper bound on the 
biclique contact-dimension of the $L^p$-metric.
Moreover, it shows a randomized construction of the polar-pair in
the $O(\log n)$-dimensional $L^p$-metric, for $p>2$,
thus implying that \closest and \bcp are equivalent for 
these $L^p$-metrics.

For algorithmic purposes, we provide a deterministic construction.
One way to derandomize the above process is to use {\em expanders}.
We show it using appropriate codes.

\begin{theorem}\label{thm:Codes}
For any $p> 2$, let $\zeta=2^{p-3}$.
There exist two sets $\left|A\right|=\left|B\right|=n$
of vectors in $\reals^{d}$, where $d=2\alpha \log_2{n}$, for some constant $\alpha \ge 1$,
such that the following holds.
\begin{enumerate}
\item[1.] For all distinct $u,u^\prime\in A$, $\|u-u^\prime\|_p>
  \left((\zeta +1/2)d\right)^{1/p}$.
\item[2.] For all distinct $v,v^\prime\in B$, $\|v-v^\prime\|_p>
  \left((\zeta +1/2)d\right)^{1/p}$.
\item[3.] For all $u\in A,\, v\in B$, $\|u-v\|_p=d^{1/p}$.
\end{enumerate}
Moreover, there exists a deterministic algorithm that
outputs $A$ and $B$ in time $\widetilde O(n)$.
\end{theorem}
\begin{proof}
In literature, we note that for any constant $\delta>0$, there is an
explicit binary code of (some) constant {\em relative rate} and {\em relative distance} at
least $\frac{1}{2}-\delta$ and the entire code can be listed in
quasilinear time with respect to the size of the code (see Appendix
E.1.2.5 from \cite{G09}, or Justesen codes \cite{J72}). 
To be more specific, we can construct in $O(n\log^{O(1)}n)$-time
a set $C\subseteq\{-1,1\}^{d'}$ 
such that 
(1) $|C| = n$, 
(2) $d'=d/2=\alpha \log_2{n}$ for some constant $ \alpha\ge 1$
and (3) for every two vectors $x,y\in C$, $x$ and $y$ differ on at least
$\left(\frac{1}{2}-\delta\right)d'$ coordinates, for some
constant $\delta \in\left(0,\frac{1}{4}-\frac{1}{2^p}\right)$.

We construct the sets $A$ and $B$ as subsets of
$\{-1,0,1\}^{d}$. For every $i\in [n]$, the
$i^{\text{th}}$ point of $A$ is given by the concatenation of the
$i^{\text{th}}$ point of $C$ with $0^{d'}$. Similarly,
the $i^{\text{th}}$ point of $B$ is given by the concatenation of
$0^{d'}$ with the $i^{\text{th}}$ point of $C$ (note the
reversal in the order of the concatenation).
In particular, points in $A$ and $B$ are of the form 
$(x_i,{\vec{0}})$ and $({\vec{0}},x_i)$, respectively,
where $x_i$ is the $i^{\text{th}}$ point in $C$ and 
${\vec{0}}$ is the zero-vector of length $\alpha \log_2  n$.

First, consider any two points in the same set, 
say $u,u'\in A$ (resp., $v,v'\in B$).
We have from the distance of $C$ that on at least 
$\left(\frac{1}{2}-\delta\right)d'$ coordinates
the two points differ by $2$,
thus implying that their $L^p$-distance is at least
\begin{align*}
\left(\left(\frac{1}{2}-\delta\right){d'}2^p\right)^{1/p}
> \left(\left(\frac{1}{4}+\frac{1}{2^p}\right){d'}2^p\right)^{1/p}
= \left(\left(2^{p-3}+\frac{1}{2}\right){d}\right)^{1/p}.
\end{align*}
This proves the first two items of the theorem.
Next we prove the third item.
Consider any two points from different sets, say $u\in A$ and $v\in B$.
It is easy to see from the construction that 
$u$ and $v$ differ in every coordinate by exactly $1$.
Thus, the $L^p$-distance between any two points from
different set is exactly
\[
\left(2d'\right)^{1/p}=d^{1/p}.\qedhere
\]
\end{proof}

%% file: Connection.tex
\section{Fine-Grained Complexity of \closest in $L^{\infty}$}
\label{sec:connection-to-closest-pair}

In this section, we prove the quadratic-time hardness of 
\closest in the $L^{\infty}$-metric.
%
Our reduction is from the {\em Orthogonal Vectors} problem (\ovec),
which we phrase as follows.
Given a pair of collections of vectors $U,W\subseteq\{0,1\}^d$,
the goal is to find a pair of vectors $u\in U$ and $w\in W$
such that $(u_i,w_i)\in \{(0,0),(0,1),(1,0)\}$ for all $i\in[d]$.
Throughout, we denote by $n$ the total number of vectors
in $U$ and $W$.

\subsection{Reduction}
\label{sec:base-reduction}

Let $U,W\subseteq\{0,1\}^d$ be an instance of \ovec.
We may assume that $U$ and $W$ have no duplicates.
Otherwise, we may sort vectors in $U$ (resp., $W$)
in lexicographic order and then sequentially remove duplicates;
this preprocessing takes $O(dn\log n)$-time.

We construct a pair of sets $A,B\subseteq\reals^d$ of \bcp
from $U,W$ as follows.
For each vector $u\in U$ (resp., $w\in W$),
we create a point $a\in A$ (resp., $b\in B$)
such that 

\begin{minipage}{0.48\textwidth}
$$a_j=\begin{cases}
0& \text{if }u_j=0,\\
2& \text{if }u_j=1.
\end{cases}$$
\end{minipage} 
\begin{minipage}{0.48\textwidth}
$$b_j=\begin{cases}
1& \text{if }w_j=0,\\
-1& \text{if }w_j=1.
\end{cases}$$
\end{minipage} \vspace{0.1cm}

Observe that, for any vectors $a\in A$ and $b\in B$, 
$|a_j-b_j|=3$ only if $u_j=w_j=1$; otherwise,
$|a_j-b_j|=1$.
It can be seen that $\|a-b\|_p=d$ if and only if 
their corresponding vectors $u\in U$ and $w\in W$
are orthogonal.
Thus, this gives an alternate proof for the
quadratic-time hardness of $\bcp$ under OVH.

\subsection{Analysis}
\label{sec:max-norm-closest-pair}

Here we show that the reduction in Section~\ref{sec:base-reduction}
rules out both exact and $2$-approximation algorithm 
for \closest in $L^{\infty}$ 
that runs in subquadratic-time 
(unless OVH is false).
That is, we prove Theorem~\ref{thm:closest-hardness-max-norm},
which follows from the theorem below.


\begin{theorem}\label{thm:Hardness-ClosestPairinf}
Assuming OVH, for any $\varepsilon>0$ and $d=\omega(\log n)$,
there is no $O\left(n^{2-\varepsilon}\right)$-time algorithm
that, given a point-set $P\subseteq\reals^d$, 
distinguishes between the following two cases:
\begin{itemize}
\item There exists a pair of vectors in $P$ with $L^{\infty}$-distance one.
\item Every pair of vectors in $P$ has $L^{\infty}$-distance two.
\end{itemize}
In particular, approximating \closest in 
the $L^{\infty}$-metric to within a factor of two
is at least as hard as solving the Orthogonal Vectors problem.
\end{theorem}

\begin{proof}
Consider the sets $A$ and $B$ constructed 
from an instance of \ovec in Section~\ref{sec:base-reduction}.

First, observe that every inner pair has 
$L^{\infty}$-distance at least $2$.
To see this, consider an inner pair $a,a'\in A$.
Since all inner pairs are distinct, 
they must have at least one different coordinate, say
$a_j\neq a'_j$ for some $j\in\{1,\ldots,n/2\}$.
Consequently, $(a_j,a'_j)\in\{(0,2),(2,0)\}$,
implying that the $L^{\infty}$-distance of $a$ and $a'$
is at least $2$.
The case of an inner pair $b,b'\in B$ is similar.
Thus, any pair of vectors with $L^{\infty}$-distance less than two
must be a crossing pair $a\in A, b\in B$.

Now suppose there is a pair of orthogonal vectors $u^*\in U, w^*\in W$,
and let $a^*\in A$ and $b^*\in B$ be the corresponding vectors 
of $u^*$ and $w^*$ in the \closest instance, respectively.
Then we know from the construction that 
$(a^*_j,b^*_j)\in\{(0,1),(0,-1),(2,1)\}$
for all coordinates $j\in[n]$.
Thus, the $L^{\infty}$-distance of $a^*$ and $b^*$ must be one.

Next suppose that there is no orthogonal pair of vectors in $U \times W$.
Then every pair of vectors $(u,w)\in U\times W$  must have 
one coordinate, say $j$, such that $u_j=w_j=1$.
So, the corresponding vectors $a$ and $b$ (of $u$ and $w$,
respectively) must have $a_i=2, b_j=-1$.
This means that $a$ and $b$ have $L^{\infty}$-distance at least three.
(Note that there might be an inner pair with $L^{\infty}$-distance two.)
Therefore, we conclude that every pair of points in $A\cup B$
has $L^{\infty}$-distance at least two.
\end{proof}

%% file: Conclusion.tex
\section{Conclusion and Discussion}
\label{sec:Conclusion}

We have studied the sphericity and contact dimension of 
the complete bipartite graph in various metrics. 
We have proved lower and upper bounds on these measures for some metrics.
However, biclique sphericity and biclique contact dimension in 
the $L^1$-metric remains poorly understood as we 
are unable to show any strong upper or lower bounds. 
However, we believe that both $L^1$ and $L^2$ metrics 
have linear upper and lower bounds. 
To be precise, we raise the following conjecture:

\begin{conjecture}[$L^1$-Biclique Sphericity Conjecture]
$$\bisph(L^1)=\Omega(n).$$
\end{conjecture}

We have also shown conditional lower bounds for the Closest Pair problem in the $L^p$-metric, for all $p\in\reals_{> 2}\cup\{\infty\}$, by using polar-pair of point-sets. 
However, it is unlikely that our techniques could get to the regime of
$L^2$, $L^1$, and $L^0$, which are popular metrics.
An open question is thus whether there exists an alternative technique
to derive a lower bound from OVH to the Closest Pair problem
for these metrics.
The answer might be on the positive side, i.e., 
there might exist an algorithm that performs well in the $L^2$-metric 
because there are more tools available, e.g.,
Johnson-Lindenstrauss' dimension reduction. 
Thus, it is possible that there exists a strongly subquadratic-time
algorithm in the $L^2$-metric.
This question remains 
an outstanding open problem. 

%% file: Acknowledgement.tex
\paragraph*{Acknowledgements}
We are grateful to the anonymous reviewers for their detailed comments and for identifying a gap in the proof of Theorem~\ref{lem:NoDistL1} and helping us fix it (and even strengthen it).
We would like to thank Aviad Rubinstein for sharing with us \cite{R17} and also for pointing out a mistake in an earlier version of the paper.
We would like to thank Petteri Kaski and Rasmus Pagh for useful discussions 
and also for pointing out the reference \cite{AW15}. 
We would like to thank Eylon Yogev and Amey Bhangale for some preliminary discussions.
We would like to thank Uriel Feige for a lot of useful comments and discussions.
Finally, we would like to thank  Inbal Livni Navon, Orr Paradise, and Roei Tell for 
helping us improve the presentation of the paper.

Roee David is supported by the ISF grant no. 621/12 and by the I-CORE Program grant no. 4/11.
Karthik C.~S. is supported by the ERC-StG grant no. 239985 and ERC-CoG grant no. 772839.
Bundit Laekhanukit is partially supported by ISF grant no. 621/12, I-CORE grant no. 4/11
and by the DIMACS/Simons Collaboration on Bridging Continuous and Discrete Optimization through NSF grant no. CCF-1740425.
Parts of the work were done while all the authors were at the Weizmann Institute of Science,
and some parts were done while the third author was visiting the Simons Institute for the Theory of Computing.

%% file: ell2.tex
\section{Geometric Representation of Biclique in $L^2$}
\label{sec:L2-dim}
\label{sec:L2-dim:lower-bound}

In this section we prove a lower bound on $\bisph(L^2)$ of $(n-3)/2$
 using spectral analysis.

%

\begin{theorem}
\label{thm:Lower bound on the embedding dimension}
For every $n,d\in \integers$, and any two sets $A,B\subseteq \reals^{d}$,
each of cardinality $n$, suppose the following holds
for some non-negative real numbers $\alpha$ and $\beta$ with $\alpha>\beta$.
\begin{enumerate}
\item For every distinct $u$ and $v$ both in $A$, $\|u-v\|_2>\alpha$.
\item For every distinct $u$ and $v$ both in $B$, $\|u-v\|_2>\alpha$.
\item For  every $u$ in $A$ and $v$ in $B$, $\|u-v\|_2\le \beta$.
\end{enumerate}
Then the dimension $d$ must be at least $\frac{n-3}{2}$.
\end{theorem}

\begin{proof}
Let $\left|A\right|=\left|B\right|=n$ be \emph{arbitrary} two sets
of vectors in $\reals^{d}$ that satisfy the above conditions.
We will show that $d\ge \frac{n-3}{2}$.
First, we scale all the vectors in $A\cup B$ so that the vector with the largest $L^2$-norm  in $A\cup B$ has
 $L^{2}$-norm equal to $1$ (by this scaling, the parameters $\alpha,\beta$ are scaled as well by, say $s$. For brevity, we will write $\alpha$ for $\alpha/s$ and similarly for $\beta$.).
We modify $A$ and $B$ in two steps as follows.
First, we add one new coordinate to all of the vectors with value $K\gg1$
(to be determined exactly later)
and obtain $A_{1},B_{1}\subseteq\reals^{d+1}$.
Note that each element in the new set of vectors $A_{1}$ and $B_{1}$
has $L^{2}$-norm  roughly equal to $K$.
More specifically, the square of the $L^2$-norm is bounded between $K^2$ and $K^2+1$ and the vector with the largest $L^2$-norm in $A_1\cup B_1$ has
 $L^{2}$-norm equal to $\sqrt{K^2+1}$.

By adding to
the last coordinate of each vector $u$ in $A_{1}\cup B_{1}$
a positive  value $c_u$ smaller than $1/K$,
we can impose that all the vectors are
 with $L^{2}$-norm equal to $\sqrt{K^2+1}$.
To see this, note that if we have a vector $u_1$ in  $A_1\cup B_1$ that has $L^2$-norm equal to $K$ (namely, as small as possible), then by setting $c_{u_1}$ to satisfy 
\begin{equation}
(K+c_{u_1})^2 = K^2 +1, \label{eq:K>>1}
\end{equation}
we have that the $L^2$-norm of $u_1$ is $\sqrt{K^2+1}$.
So, any $c_{u_1}$ that solves Eq.~\ref{eq:K>>1} is smaller than $1/K$.
By assuming that $u_1$ has a larger $L^2$-norm, we would have a better bound on $c_{u_1}$.

Let $A'_{1} \cup B'_{1}$  be the set of vectors that was obtained by adding $c_u$'s as described above. Let $u,v$ be vectors in $A_{1}\cup B_{1}$ and let $u',v'$ be the corresponding  vectors in $A'_{1}\cup B'_{1}$. By definition, the following holds:
\begin{align*}
\|u-v\|_2^2\le\|u'-v'\|_2^2
=\|u-v\|_2^2+(c_u-c_v)^2
\le \|u-v\|_2^2 + 1/K^2.
\end{align*}

Hence, by choosing $K$ to satisfy $1/K^2\le \frac{\alpha^2-\beta^2}{2}$, it follows that $A'_{1} \cup B'_{1}$  satisfies the conditions of the theorem with $\alpha'=\alpha$ and $\beta'=\sqrt{\beta^2+\frac{\alpha^2-\beta^2}{2}}< \alpha'$. Again, for brevity, we refer to  $\alpha'$ as $\alpha$ and $\beta'$ as $\beta$.

Given $A'_{1},B'_{1}\subseteq\reals^{d+1}$, let $a_{1},a_{2},\ldots,a_{n}$
be the vectors from $A'_{1}$, and $b_{1},b_{2},\ldots,b_{n}$ be the vectors
from $B'_{1}$.
Consider the following matrix in $\reals^{2\left(d+1\right)\times 2n}$:

\begin{equation}
M=\left(\begin{array}{cc}
a_{1},a_{2},\ldots,a_{n} & b_{1},b_{2},\ldots,b_{n}\\
b_{1},b_{2},\ldots,b_{n} & a_{1},a_{2},\ldots,a_{n}
\end{array}\right)\,\label{eq:Symmetric swap}
\end{equation}

Define the set $A_{2}$ to be the first $n$ column vectors of $M$
and  $B_{2}$ to be the last $n$ column vectors of $M$.
Note that $A_{2}\cup B_{2}\subseteq\reals^{2\left(d+1\right)}$
and that the pair of point-sets satisfy the conditions of the theorem with $\alpha''=\sqrt{2}\alpha'>\sqrt{2}\beta'=\beta''$.
Consider the inner product matrix $M^{T}M\in\reals^{2n\times 2n}$
written in a block matrix form as follows:
$$M^{T}M=cI_{2n\times 2n}+\left(\begin{array}{cc}
M_{1,1} & M_{1,2}\\
M_{2,1} & M_{2,2}
\end{array}\right)
,$$
where $M_{1,1},M_{1,2},M_{2,1},M_{2,2}\in\reals^{n\times n}$
and $c$ is such that the matrix $\left(\begin{array}{cc}
M_{1,1} & M_{1,2}\\
M_{2,1} & M_{2,2}
\end{array}\right)$ has the value $0$ on the diagonal elements
(recall that all the vectors have the same $L^{2}$-norm).
By the definition of $M$ (see Eq.~\ref{eq:Symmetric swap}),
one can check that the following hold.
\begin{enumerate}
\item The matrices $M_{1,1},M_{1,2},M_{2,1},M_{2,2}$ are all symmetric:
for $M_{1,1},M_{2,2}$ it follows since $M^{T}M$ is a symmetric matrix, and
for $M_{1,2},M_{2,1}$ it follows by the way $M$ was defined;
see Eq.~\ref{eq:Symmetric swap}.
\item $M_{1,1}=M_{2,2}$. This follows by Eq.~\ref{eq:Symmetric swap}.
\item $M_{1,2}=M_{2,1}$. This follows since $M_{1,2}=M_{2,1}^{T}=M_{2,1}$.
Here the first equality follows since $M^{T}M$ is a symmetric matrix,
and the last equality follows by item~1.
\end{enumerate}
Hence, we can write
$
M^{T}M=cI_{2n\times 2n}+\left(\begin{array}{cc}
M_{1,1} & M_{1,2}\\
M_{1,2} & M_{1,1}
\end{array}\right)\,.
$
In the rest of the proof, we analyze some of the eigenvectors of $M^{T}M$.
To this end, we consider the matrix $M_{1,1}-M_{1,2}$.
Since both $M_{1,1}$ and $M_{1,2}$ are symmetric,
we have that $M_{1,1}-M_{1,2}$ is symmetric and has real eigenvalues.
Moreover, by the conditions of the theorem, it holds that $M_{1,1}-M_{1,2}$
is strictly negative (i.e., all the entries of the matrix are negative).
This follows because  all the vectors have the same $L^{2}$-norm.
Let $x_{1},x_{2},\ldots,x_{n}$ be the eigenvectors
of $M_{1,1}-M_{1,2}$ with eigenvalues $\lambda_{1},\lambda_{2},\ldots,\lambda_{n}$.
By the Perron\textendash Frobenius Theorem it follows that $\lambda_{1}$
is strictly smaller than $\lambda_{2},\lambda_{3},\ldots,\lambda_{n}$.

Let $x_{i}\in\reals^{n}$ be an eigenvector of $M_{1,1}-M_{1,2}$
with eigenvalue $\lambda_{i}$.
Then the following holds.
\begin{align*}
\left(\begin{array}{cc}
M_{1,1} & M_{1,2}\\
M_{1,2} & M_{1,1}
\end{array}\right)\left(\begin{array}{c}
x_{i}\\
-x_{i}
\end{array}\right)  &=\left(\begin{array}{c}
\left(M_{1,1}-M_{1,2}\right)x_{i}\\
-\left(M_{1,1}-M_{1,2}\right)x_{i}
\end{array}\right)\\
 &=\left(\begin{array}{c}
\lambda_{i}x_{i}\\
-\lambda_{i}x_{i}
\end{array}\right)\\
&=\lambda_{i}\left(\begin{array}{c}
x_{i}\\
-x_{i}
\end{array}\right).
\end{align*}

Hence, the vectors $\left(\begin{array}{c}
x_{1}\\
-x_{1}
\end{array}\right),\left(\begin{array}{c}
x_{2}\\
-x_{2}
\end{array}\right),\ldots,\left(\begin{array}{c}
x_{n}\\
-x_{n}
\end{array}\right)$ are eigenvectors of $\left(\begin{array}{cc}
M_{1,1} & M_{1,2}\\
M_{1,2} & M_{1,1}
\end{array}\right)$ with eigenvalues
$\lambda_{1},\lambda_{2},\ldots,\lambda_{n}$.
The operation of adding $cI_{2n\times 2n}$ to $\left(\begin{array}{cc}
M_{1,1} & M_{1,2}\\
M_{1,2} & M_{1,1}
\end{array}\right)$ shifts the eigenvalues of $M^{T}M$ to
$\lambda_{1}+c,\lambda_{2}+c,\ldots,\lambda_{n}+c$.

Since $M^{T}M$ is a positive semidefinite matrix, 
$\lambda_{1}+c,\lambda_{2}+c,\ldots,\lambda_{n}+c\ge0$.
More specifically, $\lambda_{1}+c\ge0$ and
$\lambda_{2}+c,\ldots,\lambda_{n}+c > 0$
(since $\lambda_{1}<\lambda_{2},\lambda_{3}\ldots,\lambda_{n}$).
It follows that $M^{T}M$ has at least $n-1$ positive eigenvalues.
Hence, the rank of $M^{T}M$ is at least $n-1$.
By standard linear algebra arguments, it holds that the rank of $M$
is at least the rank of $M^{T}M$,
and the rank of $M$ is at most $2\left(d+1\right)$.
That is,
\[
2(d+1) \geq \mathrm{rank}(M) \geq \mathrm{rank}(M^TM) \geq n - 1.\qedhere
\]
\end{proof}